\documentclass[twoside,a4paper]{article}
\usepackage{amssymb}
\usepackage{color}
\usepackage{graphicx}
\usepackage{subfigure}
\usepackage{cite}

\usepackage[english]{babel}

\usepackage{amsmath}
\usepackage{amsthm}
\usepackage{bbm}

\newtheorem{definition}{Definition}[section]
\newtheorem{proposition}{Proposition}[section]

\newtheorem{remark}{Remark}[section]
\newtheorem{corollary}{Corollary}[section]

\begin{document}

\title{Performances and robustness of quantum teleportation with identical particles}
\author{Ugo Marzolino$^{1,2}$ and Andreas Buchleitner$^2$\\
\\
\small 
$^1$Univerza v Ljubljani, Jadranska 19, SI-1000 Ljubljana, Slovenija \\
\small
$^2$Albert-Ludwigs-Universit\"at Freiburg, Hermann-Herder-Stra\ss e 3, 79104 Freiburg, Deutschland
}
\date{\null}

\maketitle

\begin{abstract}
When quantum teleportation is performed with truly identical massive particles, indistinguishability allows us to teleport addressable degrees of freedom which do not identify particles, but e.g. orthogonal modes. The key resource of the protocol is a state of entangled modes, but the conservation of the total number of particles does not allow for perfect deterministic teleportation unless the number of particles in the resource state goes to infinity. Here, we study the convergence of teleportation performances in the above limit, and provide sufficient conditions for asymptotic perfect teleporation. We also apply these conditions to the case of resource states affected by noise.
\end{abstract}

\section{Introduction}

Information theory has been successfully extended to the quantum domain, where information processing is implemented with systems ruled by the laws of quantum mechanics \cite{NielsenChuang,Horodecci,Buhrman2010}. Most theoretical work focuses on distinguishable particles, namely particles that can be unambiguously identified at any moment. However, quantum mechanics predicts that identical particles cannot be in general distinguished \cite{Landau,Messiah}. On the other hand, many experiments and proposals actually employ identical particles \cite{Ionicioiu2002,Riebe2004,Barrett2004,Negretti2011,Rohling2012,
Underwood2012,Inaba2014,Byrnes2015}, implementing protocols derived for distinguishable particles. The usual method is to fix some degrees of freedom, e.g. position values, to unambiguously label and characterize each particle \cite{Herbut2001,Herbut2006,Tichy2013}. These degrees of freedom cannot be further manipulated, whereas additional degrees of freedom must carry the relevant information to be processed.

Nevertheless, if one aims to build an integrated architecture, it can be required to exploit all the manipulable degrees of freedom, in order to encode the desired processing, with no way to label and distinguish particles. For instance, if identical particles are distinguished and labelled by means of spatial localization, only the other degrees of freedom, such as spin, polarization or hyperfine levels, are accessible within the framework of distinguishable particles. It is however hard to localize particles in small systems, an the exploitation of all the degrees of freedom, including the spatial one, is convenient to increase the computing power of a device without increasing its size.
For these reasons, we start a careful analysis of quantum information protocols, focusing on the generalization of quantum teleportation \cite{Bennett1993}, based on identical particles.

Key resources for teleportation and other quantum protocols are entangled states. The latter are quantum states that exhibit correlations between subsystems, which cannot be explained by a classical probability theory. Since identical particles cannot be individually addressed in general, it is not meaningful to define correlations between particles.
Instead, we apply the approach developed in \cite{Zanardi2001,Zanardi2004,Barnum2004,Narnhofer2004,Benatti2010,Benatti2012,Benatti2012-2,Benatti2014,Benatti2014-2}, where entanglement is defined via non-classical correlations between commuting subalgebras of observables, i.e. partitions of all the observables in groups that are physically addressable without mutual disturbance. Thus, correlations are defined by means of the possibility to write expectation values of physically accessible observables in terms of a classical probability distribution. This is a very general and powerful approach that recovers the standard definition of entanglement for distinguishable particles, and suitably generalizes the concept of entanglement to identical particles. We shall focus on massive particles at non-relativistic energies, like atoms and constituents of condensed matter systems, which satisfy the conservation of the total number of particles, mathematically described by a superselection rule \cite{Bartlett2007}. Based on this property, entanglement of identical particles was proven to exhibit markedly different features from those of photons and distinguishable particles, such as much simpler detectability \cite{Benatti2012,Benatti2012-2,Benatti2014}, a higher robustness against noise \cite{Benatti2012-2,Marzolino2013,Benatti2014}, and the geometry of entangled states drawn in the whole space of quantum states \cite{Benatti2012-2,Benatti2014}.

The properties of entanglement have been studied in Fermionic superconducting systems \cite{Zanardi2002-2,Vedral2004}, electrons in low-dimensional semiconductors \cite{Buscemi2007}, and bosonic ultracold gases \cite{Anders2006,Goold2009,Benatti2010,Benatti2011,Argentieri2011}, and exploited in several applications, such as quantum data hiding \cite{Verstraete2003}, teleportation \cite{Marinatto2001,Schuch2004,Heaney2009,Molotkov2010,Marzolino2015}, Bell's inequalities \cite{Ashhab2009,Heaney2010}, dense coding \cite{Heaney2009}, and quantum metrology \cite{Benatti2010,Benatti2011,Argentieri2011,Benatti2014}. Since quantum teleportation is a primitive for scalable quantum computers \cite{Gottesman1999}, and plays a fundamental role in measurement-based quantum computation \cite{Gross2007,Chiu2013}, we will focus on this special topic hereafter.

In the original teleportation protocol \cite{Bennett1993} one agent, Alice, wants to teleport an arbitrary, perhaps unknown, state to another agent, Bob. Alice owns the state to be teleported and a share of a resource state, and Bob owns the remaining part of the resource state. The algorithm of the standard teleportation is the following: \emph{i)} Alice performs a projective measurement onto the basis of maximally entangled states of her states; \emph{ii)} Alice sends Bob the result of the measurement; \emph{iii)} Bob performs a suitable operation on his state, conditioned on the message he got from Alice. In the setting of distinguishable particles, if the shared state is a pure, maximally entangled state, Bob ends up with a state identical to the initial state to be teleported, while the initial state has been transformed by the measurement.
The teleportation can also be applied to a part of an entangled state. In this case, Alice teleports the state of a system entangled with another one, and the initial entanglement is perfectly swapped with Bob at the end of the protocol. This application is called entanglement swapping, and can be useful for sharing entanglement at long distances.

In \cite{Marzolino2015}, we developed a teleportation protocol implementable with identical massive particles. We discussed its performances  with some physically interesting resource states: non-entangled states, maximally entangled states, $SU(2)$ (or atomic) coherent states, and ground states of the double well potential with two-body interactions. These two latter resource states are of particular interest, since they can be experimentally prepared with nowadays' technology. We observed that our protocol cannot perfectly teleport any general quantum state. Furthermore, we proved that this is a general feature for any teleportation protocol with identical massive particles which is implemented by local operations and classical communication between the agents. Nevertheless, perfect teleportation can be achieved when the number of particles in the resource state increases to infinity. In the present paper, we study the convergence of teleportation performances in this limit. We also apply our analysis to physically relevant resource states and to the robustness of teleportation performances in the presence of noise.

We will present the basic definition of entanglement in section \ref{entanglement} and the aforementioned teleportation protocol in section \ref{protocol}. In section \ref{perf}, we shall show that the efficiency grows with the number of particles in the resource state. From a detailed analysis of the resource state, we shall give simple sufficient conditions for the resource state to provide perfect teleportation when the number of particles goes to infinity. These conditions recover and generalize the resource states studied in \cite{Marzolino2013}. In section \ref{robustness}, we shall discuss the robustness of the teleportation performances, when the resource state is affected by noise. We shall sum up our conclusions in section \ref{discussions}.

\section{Entanglement} \label{entanglement}

Before discussing the teleportation protocol, we introduce the algebraic formalism on which the notion of entanglement is based. This formalism generalizes that of distinguishable particles and can be applied unambiguously to identical particles \cite{Zanardi2001,Zanardi2004,Barnum2004,Narnhofer2004,Benatti2010,Benatti2012,Benatti2012-2,Benatti2014,Benatti2014-2}. Physical observables are self-adjoint elements of a $C^\star$-algebra, that can be represented as the algebra ${\cal B}({\cal H})$ of bounded operators on a Hilbert space ${\cal H}$ \cite{BratteliRobinson,EspositoMarmoSudarshan}. A state of the system is a positive functional on ${\cal{B}({\cal H})}$, namely a linear map $\omega:{\cal{B}({\cal H})}\to\mathbbm{C}$ with $\omega(A^\dag A)\geqslant 0$ and with the normalization $\omega(\mathbbm{1})=1$.

We now identify two subsystems by means of two commuting subalgebras of operators.

\begin{definition}[Algebraic bipartition]
An \emph{algebraic bipartition} of ${\cal B}({\cal H})$ is any pair
$({\cal A}_1, {\cal A}_2)$ of commuting subalgebras,
${\cal A}_1, {\cal A}_2\subset {\cal B}({\cal H})$.
\end{definition}

\noindent
Any element of ${\cal A}_1$ commutes with any element of ${\cal A}_2$, $[{\cal A}_1, {\cal A}_2]=0$. The notion of locality lies in the commutativity of the subalgebras, ensuring that the results of a joined measurement of an observable in ${\cal A}_1$ and of an observables in ${\cal A}_2$ do not depend on the ordering. Therefore, subsystems are defined via observables that are individually and unambiguously addressable.

\begin{definition}[Local operators]
An operator is said to be \emph{local} with respect to the bipartition $({\cal A}_1, {\cal A}_2)$, if it is the product $A_1 A_2$ of an operator 
$A_1$ of ${\cal A}_1$ and another $A_2$ in ${\cal A}_2$.
\end{definition}

\noindent
We are now ready to define quantum correlated states with respect to a given algebraic bipartition.

\begin{definition}[Entangled states]
A state $\omega$ is said to be \emph{separable} with respect to the bipartition $({\cal A}_1, {\cal A}_2)$ if the expectation of any local operator $A_1 A_2$ can be decomposed into a linear convex combination of
products of local expectations:

\begin{equation}
\omega(A_1 A_2)=\sum_k\lambda_k \, \omega_k^{(1)}(A_1)\omega_k^{(2)}(A_2),\qquad \lambda_k\geq 0,\qquad \sum_k\lambda_k=1,
\end{equation}

\noindent
with $\omega_k^{(1)}$ and $\omega_k^{(2)}$ being \emph{bona fide} states of the system. Otherwise, the state is \emph{entangled}.
\end{definition}

We now specialize these definitions to $N$ bosons whose single particle Hilbert space has finite dimension $M$. Any state $\omega$ can be represented by a positive operation $\rho\in{\cal B}({\cal H})$ \cite{BratteliRobinson}, such that

\begin{equation}
\omega(A)=\textnormal{tr}(\rho A), \qquad \textnormal{tr} \, \rho=1.
\end{equation}

\noindent
{\it Pure states} $\rho=\rho^2$ are projectors, isomorphic to elements in ${\cal H}$, while states $\rho\neq\rho^2$ are called {\it mixed}.

Subalgebras of single particle observables in the above definition covers the usual definition of entanglement of distinguishable particles \cite{Werner1989}. At this point indistinguishability comes into play: in systems of identical particles, there is no subalgebras of physical observables acting on an individual particle. Thus, our approach provides the required generalisation.

The formalism of second quantization is more convenient for identical particles. Let us introduce creation and annihilation operators $a^\dagger_j,a_j$, $j=1, 2,\ldots,M$, of $M$ modes, that satisfy the Bosonic commutation relations, $[a_j,\,a^\dagger_l]=\delta_{jl}$. The many-body Hilbert space ${\cal H}_N$ of the system is spanned by the Fock states,

\begin{equation} \label{fock.states}
|k_1\rangle\otimes|k_2\rangle\otimes\cdots\otimes|k_M\rangle= \frac{(a_1^\dagger)^{k_1}\, (a_2^\dagger)^{k_2}\, \cdots\, (a_M^\dagger)^{k_M}\,|0\rangle}{\sqrt{k_1!\, k_2!\cdots k_M!}},
\end{equation}

\noindent
where $|0\rangle$ is called {\it vacuum state}, and the integer $k_j$ is the occupation number of the $j$-th mode, such that $\sum_{j=1}^M k_j=N$. Fixing the number of particles constrains the Hilbert space to ${\cal H}_N=P_N\mathbbm{F}$, where $\mathbbm{F}$ is the Fock space, spanned by the unconstrained Fock states (\ref{fock.states}), and $P_N$ is the projector on the eigenspace with $N$ particles. This is a constraint to the linearity of the Hilbert space, formalized by a superselection rule \cite{Bartlett2007}, where superpositions of different total number of particles are forbidden. The physical motivation of the superselection rule stems from the impossibility to create massive particles at non-relativistic energies. We use the notation of tensor products, to express mode-partitions: e.g. $a_1^\dag a_2^\dag|0\rangle=a_1^\dag|0_1\rangle\otimes a_2^\dag|0_2\rangle$. We have to keep in mind that this tensor product structure is constrained by the conservation of the number of particles.

The norm-closure of the set of polynomials in all creation and annihilation operators,
$\{a^\dagger_j,\, a_j\}$, $j=1,2,\dots, M$ is the algebra
${\cal B}({\cal H}_N)$.
\footnote{The algebra ${\cal{B}({\cal H})}$ is generated by differentiation of the so-called Weyl operators. Any function of the creation and annihilation operators is obtained by a proper differentiation of the Weyl operators \cite{BratteliRobinson,EspositoMarmoSudarshan}.}.
We define the bipartition of this algebra by splitting the set of creation and annihilation operators into two disjoint sets
$\{a_j^\dagger,a_j \, | \, j=1,2\dots,m\}$ and 
$\{a_l^\dagger,a_j \, | \, j=m+1,m+2,\dots,M\}$. The norm-closure of all polynomials in the creation and annihilation operators of the first (second) set is the subalgebra ${\cal A}_1$ (${\cal A}_2$). According to the previous definition, a pure state is $({\cal A}_1,{\cal A}_2)$-separable if and only if

\begin{equation}
|\psi\rangle={\cal P}(a^\dagger_1, \dots, a^\dagger_m)\cdot 
{\cal Q}(a^\dagger_{m+1},\ldots ,a^\dagger_M)\ |0\rangle, \label{sep}
\end{equation}

\noindent
and mixed $({\cal A}_1,{\cal A}_2)$-separable states are convex combinations of pure $({\cal A}_1,{\cal A}_2)$-separable states. See \cite{Benatti2010,Benatti2012,Benatti2014} for a detailed analysis.

Beyond its definition, entanglement can be quantified by means of measures of information carried by the subsystems. Minimal axioms fulfilled by these measures are \cite{Horodecci}: they are positive and vanish for all separable states, they are invariant under local unitary operations on the state, they do not increase under local operations and classical communication. The most easily computable measure of entanglement for mixed states is the so-called negativity \cite{Vidal2002}, based on the partial transposition operation, denoted by $^T$ \cite{Peres1996},

\begin{equation} \label{negativity}
{\cal N}(\rho)=\frac{\textnormal{tr}\sqrt{(\rho^T)^2}-1}{2}.
\end{equation}

\noindent
In general, negativity can be zero for some entangled states. Nevertheless, the set of entangled states of identical massive bosons that are not detected by the negativity has measure zero, and is the null set in the special case of two-mode states where ${\cal N}(\rho)=\sum_{k\neq j}|\rho_{k,j}|/2$ \cite{Benatti2012}.

Previous proposals for implementing quantum teleportation \cite{Marinatto2001,Molotkov2010} and general QIP \cite{Ionicioiu2002,Negretti2011,Rohling2012,Underwood2012,Inaba2014} with identical particles are based on states of identical particles that can be distinguished by the consumption of some degrees of freedom, as mentioned in the introduction. In the following, we analyse performances of a teleportation protocol that uses states of many-particles with a single two-mode degree of freedom. Thus, it will not be possible to distinguish identical particles.

\section{Teleportation protocol} \label{protocol}

In this section, we describe the teleportation protocol \cite{Marzolino2015}. In the standard protocol with distinguishable particles, each agent owns one particle. When teleportation is implemented with identical particles, each agent owns addressable subsystems which are not particles but rather modes \cite{Wurtz2009,Gross2010,Riedel2010}. We aim to teleport the state of one mode of a two-mode state $|\psi_{12}\rangle$, with the help of a two-mode shared resource state $\rho_{34}$. The labels $1,2,3,4$ number the modes. The initial global state is $|\psi_{12}\rangle\langle\psi_{12}|\otimes\rho_{34}$, where

\begin{equation} \label{initial.state}
|\psi_{12}\rangle=\sum_{k=0}^N c_k|k\rangle_1\otimes|N-k\rangle_2, \qquad \sum_{k=0}^N|c_k|^2=1,
\end{equation}

\noindent
and $\rho_{34}$ is a general state of $\nu$ two-mode particles

\begin{equation} \label{res}
\rho_{34}=\sum_{k,l=0}^\nu\left(\rho_{34}\right)_{k,l}|k\rangle_3{\,}_3\langle l|\otimes|\nu-k\rangle_4{\,}_4\langle \nu-l|.
\end{equation}

First, Alice performs a complete projective measurement on the second and the third mode. This means that when Alice measures the system, the latter is projected onto a state that depends on the measurement outcome $(l,\lambda)$. Moreover, the sum of all the projectors is the identity matrix. The projectors are $P_{23}^{(l,\lambda)}=|\phi_{23}^{(l,\lambda)}\rangle\langle\phi_{23}^{(l,\lambda)}|$, where

\begin{align}
& \label{states.meas} |\phi_{23}^{(l,\lambda)}\rangle=\sum_{k=\max\{0,-l\}}^{\min\{N,\nu-l\}}\frac{e^{2\pi i\frac{\lambda k}{{\cal C}_l}}}{\sqrt{\cal C}_l}|N-k\rangle_2\otimes|k+l\rangle_3, \\
& l\in\{-N,-N+1\dots,\nu\}, \qquad \lambda\in\{0,1,\dots,{\cal C}_l-1\}, \\
& {\cal C}_l=
\begin{cases}
N+l+1 & \mbox{if } -N\leqslant l\leqslant 0 \\
N+1 & \mbox{if } 0\leqslant l\leqslant \nu-N \\
\nu-l+1 & \mbox{if } \nu-N\leqslant l\leqslant \nu
\end{cases}.
\end{align}

\noindent
Notice that each state (\ref{states.meas}) has a fixed number of particles, in accordance with the conservation of the total particle number, but that different $|\phi_{23}^{(l,\lambda)}\rangle$ do not have all the same amount of entanglement. In particular, states (\ref{states.meas}) with $0\leqslant l\leqslant\nu-N$ have the same amount of entanglement as the maximally entangled states of $N$ two-mode particles, in terms of the entanglement measures discussed in \cite{Benatti2012,Benatti2012-2,Benatti2014}. States (\ref{states.meas}) with other values of $l$ have less entanglement. Notice also that it is not possible to write a complete projective measurement with states that preserve the total number of particles and have all the same entanglement \cite{Marzolino2015}.

If Alice measures the outcome $(l,\lambda)$, the state changes into

\begin{align} \label{op.telep}
& \mathbbm{1}_1\otimes P_{23}^{(l,\lambda)}\otimes\mathbbm{1}_4\big(|\psi_{12}\rangle\langle \psi_{12}|\otimes\rho_{34}\big)\mathbbm{1}_1\otimes P_{23}^{(l,\lambda)}\otimes\mathbbm{1}_4=\sum_{k=\max\{0,-l\}}^{\min\{N,\nu-l\}}\left(\rho_{34}\right)_{k+l,j+l}c_k\bar c_j\cdot \nonumber \\
& \cdot\frac{e^{2\pi i\lambda\frac{j-k}{{\cal C}_l}}}{{\cal C}_l}|k\rangle_1{\,}_1\langle j|\otimes|\phi_{23}^{(l,\lambda)}\rangle\langle\phi_{23}^{(l,\lambda)}|\otimes|\nu-k-l\rangle_4{\,}_4\langle\nu-j-l|,
\end{align}

\noindent
where $\mathbbm{1}_j=\sum_{k\geqslant 0}|k\rangle_j{\,}_j\langle k|$ is the identity operator on the $j$-th mode. Alice sends Bob the outcome $(l,\lambda)$ of her measurement, via a classical channel, and subsequently Bob applies the operation $V_4^{(l,\lambda)}\rho (V_4^{(l,\lambda)})^\dag$ to the fourth mode, where the operator

\begin{equation} \label{V_4}
V_4^{(l,\lambda)}=\sum_{k=\max\{0,-l\}}^{\min\{N,\nu-l\}} e^{2\pi i\frac{\lambda k}{{\cal C}_l}}|N-k\rangle_4{\,}_4\langle \nu-k-l|
\end{equation}

\noindent
is fully consistent with the total number superselection rule as shown in \cite{Marzolino2015}.

Since quantum mechanical measurements are probabilistic events \cite{Landau,Messiah,NielsenChuang}, Alice obtains the measurement outcome $(l,\lambda)$ with probability $p_{(l,\lambda)}$, and the consequent final state is $\rho_{14}^{(l,\lambda)}$, where

\begin{align} \label{telep.term}
p_{(l,\lambda)}\rho_{14}^{(l,\lambda)}\equiv & \, \textnormal{tr}_{23}\Big[\mathbbm{1}_1\otimes P_{23}^{(l,\lambda)}\otimes V_{4}^{(l,\lambda)}\big(|\psi_{12}\rangle\langle\psi_{12}|\otimes\rho_{34}\big)\mathbbm{1}_1\otimes P_{23}^{(l,\lambda)}\otimes(V_{4}^{(l,\lambda)})^\dag\Big] \nonumber \\
= & \sum_{k,j=\max\{0,-l\}}^{\min\{N,\nu-l\}}\left(\rho_{34}\right)_{k+l,j+l}\frac{c_k\bar c_j}{{\cal C}_l}|k\rangle_1{\,}_1\langle j|\otimes|N-k\rangle_4{\,}_4\langle N-j|,
\end{align}

\noindent
with $\textnormal{tr}\big(\rho_{14}^{(l,\lambda)}\big)=1$, and where $\textnormal{tr}_{23}$ is the trace over the second and the third mode. The average teleported state, generated by the operation $\mathcal{T}$, over all Alice's outcome is

\begin{align} \label{teleported}
\mathcal{T}\big[|\psi_{12}\rangle\langle\psi_{12}|\big]= & \sum_{l=-N}^\nu\sum_{\lambda=0}^{{\cal C}_l-1}p_{(l,\lambda)}\rho_{14}^{(l,\lambda)} \nonumber \\
= & \sum_{l=-N}^{\nu}\sum_{k,j=\max\{0,-l\}}^{\min\{N,\nu-l\}} c_k\bar c_j\left(\rho_{34}\right)_{k+l,j+l}|k\rangle_1{\,}_1\langle j|\otimes|N-k\rangle_4{\,}_4\langle N-j|,
\end{align}

The efficiency of teleportation is quantified by the average overlap between the state (\ref{initial.state}) and the teleported state (\ref{teleported}). This quantity, called {\it fidelity}, is $f=\int d\psi \langle\psi|\mathcal{T}[|\psi\rangle\langle\psi|]|\psi\rangle$ \cite{%Horodecki1999,
Horodecki1999-2}, where $d\psi$ is the uniform distribution over all pure states. Defining $c_k=r_k e^{i\varphi_k}$ with $r_k\geq 0$ and $0\leq\varphi_k<2\pi$, the uniform distribution $d\psi$ is induced by the Haar measure of the unitary group \cite{Horodecki1999-2,Wootters1990}:

\begin{equation} \label{measure}
d\psi=\frac{N!}{\pi^{N+1}}\delta\left(1-\sum_{k=0}^N r_k^2\right)\prod_{k=0}^N r_k dr_k d\varphi_k.
\end{equation}

\noindent
The fidelity is

\begin{equation} \label{fidelity}
f=\frac{2}{N+2}+\sum_{\substack{k,j=0\\ k\neq j}}^\nu\frac{\max\{0,N+1-|k-j|\}}{(N+1)(N+2)}\left(\rho_{34}\right)_{k,j}\in[0,1].
\end{equation}

A different figure of merit is the average entanglement between the first and the fourth mode in the teleported states (\ref{telep.term}). This figure of merit is relevant if the teleportation is applied for sharing entanglement at long distances. We quantify entanglement of the final states (\ref{telep.term}) with the negativity (\ref{negativity}). Thus, the average final entanglement is

\begin{equation} \label{av.ent}
E=\int d\psi\sum_{l=-N}^\nu\sum_{\lambda=0}^{{\cal C}_l-1}p_{(l,\lambda)}{\cal N}(\rho_{14}^{(l,\lambda)})=\frac{\pi}{8}\sum_{\substack{k,j=0\\ k\neq j}}^\nu\frac{\max\{0,N+1-|k-j|\}}{(N+1)}|\left(\rho_{34}\right)_{k,j}|\in\left[0,\frac{\pi}{8}\right],
\end{equation}

\noindent
The entanglement of each final state $\rho_{14}^{(l,\lambda)}$ does not depend on whether the local operation $V_4^{(l,\lambda)}$ has been performed. We notice that

\begin{equation} \label{triangle}
\frac{8E}{\pi}\geq(N+2)f-2
\end{equation}
follows from the triangle inequality of the absolute value, and from the positivity of the fidelity $f$. The upper bound of $E$ is the average entanglement over all pure initial states, $\int d\psi \, {\cal N}(|\psi\rangle\langle\psi|)=\pi N/8$, which cannot be exceeded because the teleportation protocol does not act on the first mode and cannot increase the entanglement between the first mode and the rest.

In \cite{Marzolino2015}, we proved that {\it deterministic perfect teleportation}, i.e. $f=1$ and $E=\pi N/8$, is not possible for any resource state of finitely many particles. This is a property of any general teleportation protocol performed on identical massive particles by local operations of Alice's and Bob's sides plus classical communication.

\begin{proposition} \label{imp}
Deterministic perfect teleportation is never possible for a fixed and finite number of identical particles.
\end{proposition}

\section{Teleportation performances} \label{perf}

In this section we shall discuss the teleportation performances, i.e. the fidelity and the average final entanglement. A natural reference is the teleportation performance given by any separable resource state $f_\textnormal{sep}=\frac{2}{N+2}$, $E_\textnormal{sep}=0$. Another interesting resource state is the maximally entangled state $\rho_{34}=|\phi_{34}\rangle\langle\phi_{34}|$ of $\nu$ two-mode particles \cite{Benatti2012,Benatti2012-2} with

\begin{equation} \label{max.ent.state}
|\phi_{34}\rangle=\frac{1}{\sqrt{\nu+1}}\sum_{k=0}^\nu|k\rangle_3\otimes|\nu-k\rangle_4.
\end{equation}

\noindent
For any measurement outcome with $0\leqslant l\leqslant\nu-N$, the teleported state (\ref{telep.term}) is perfectly the same as the initial state (\ref{initial.state}). These measurement outcomes occur with an overall probability $\frac{\nu-N+1}{\nu+1}$. In this sense, the resource state (\ref{max.ent.state}) provides a {\it probabilistic perfect teleportation}. The fidelity and the average final entanglement are $f_\textnormal{max \, ent}=1-\frac{N}{3(\nu+1)}$, and $E_\textnormal{max \, ent}=\frac{\pi N(3\nu-N+1)}{24(\nu+1)}$. Despite of proposition \ref{imp}, we showed in \cite{Marzolino2015} with some exemplary resource states, such as the maximally entangled state \eqref{max.ent.state}, atomic coherent states and the ground state of the double well potential with intra-well interactions, that deterministic perfect teleportation can be approached in the limit $\nu\to\infty$ with fixed $N$.

We now derive sufficient conditions for the deterministic perfect teleportation in the limit of infinitely many particles of the resource state $\nu\to\infty$, namely {\it asymptotically perfect teleportation}. These conditions generalize the examples discussed in \cite{Marzolino2015}.

\subsection{Sufficient conditions for asymptotically perfect teleportation}

The first step is to make a continuum approximation of the entries of the resource state:

\begin{equation} \label{cont.appr.mixed}
(\rho_{34})_{k,j}\simeq\omega(z,y)\frac{2}{\nu}.
\end{equation}
The new variables $z=1-2k/\nu$ and $y=1-2j/\nu$ represent the particle number imbalance between the third and the fourth mode. The factor $2/\nu$ guarantees the normalization

\begin{equation}
1=\sum_{k=0}^\nu(\rho_{34})_{k,k}=\int_{-1}^1 dz \, \omega(z,z),
\end{equation}
where we approximated the sum with an integral. We exploit this approximation in the fidelity (\ref{fidelity}) and in the average final entanglement (\ref{av.ent}). For this purpose, we extend the sums in (\ref{fidelity}) and in (\ref{av.ent}) also to the indices $k=j$ and subtract the corresponding values, computable by the identity

\begin{equation}
\sum_{k,j}\delta_{k,j}\max\big\{0,N+1-|k-j|\big\}(\rho_{34})_{k,j}=N+1.
\end{equation}
The results of the continuum approximation \eqref{cont.appr.mixed} are

\begin{align}
\label{fid.int} f\simeq & \frac{1}{N+2}+\frac{\nu}{2}\int_{-1}^1 dz\int_{-1}^1 dy\frac{\max\left\{0,N+1-|z-y|\frac{\nu}{2}\right\}}{(N+1)(N+2)} \, \omega(z,y) \\
\label{ent.int} E\simeq & -\frac{\pi}{8}+\frac{\pi\nu}{16}\int_{-1}^1 dz\int_{-1}^1 dy\frac{\max\left\{0,N+1-|z-y|\frac{\nu}{2}\right\}}{N+1} \, |\omega(z,y)|.
\end{align}

\noindent
The following propositions identify sufficient conditions for asymptotically perfect teleportation. We first focus on pure states in the continuum approximation:

\begin{equation} \label{chi34}
|\chi_{34}\rangle=\sum_{k=0}^\nu x_k|k\rangle_3\otimes|\nu-k\rangle_4, \qquad x_k\simeq\chi(z)\sqrt{\frac{2}{\nu}}, \qquad \omega(z,y)=\bar\chi(z)\chi(y),
\end{equation}

\begin{proposition} \label{prop}
If there are two real functions $\delta(\nu)\in[-1,1]$ and $\alpha(\nu)$ such that the function $\chi(z)$ is rescaled as $\chi(z)=\sqrt{\alpha} \, \zeta((z+\delta)\alpha)$ and the function $\zeta(\cdot)$ does not depend on $\nu$, then the asymptotic teleportation performances scale as

\begin{equation}
f=1-R-{\cal O}\left(\alpha(\nu)\frac{N}{\nu}\right), \qquad E=\frac{\pi N}{8}\left(1-R-{\cal O}\left(\alpha(\nu)\frac{N}{\nu}\right)\right). \label{fid.ent.asym}
\end{equation}

\noindent
The remainder $R$ is $R=o(\alpha(\nu))$ if $\zeta$ is continuous in $[\alpha(\delta-1),\alpha(1+\delta)]$, $R=o(\alpha^2(\nu)N/\nu)$ if $\zeta$ is differentiable in $[\alpha(\delta-1),\alpha(1+\delta)]$, and $R={\cal O}(\alpha^3(\nu)N^2/\nu^2)$ if $\zeta$ is twice differentiable in $[\alpha(\delta-1),\alpha(1+\delta)]$.
\end{proposition}

\begin{proof}
First, we notice that the factor $\sqrt{\alpha}$ in the scaling ensures the normalization of the function $\zeta(\cdot)$:

\begin{equation}
1=\int_{-1}^1 dz|\chi(z)|^2=\int_{\alpha(\delta-1)}^{\alpha(1+\delta)}dz'|\zeta(z')|^2,
\end{equation}
where $z'=(z+\delta)\alpha$. The factor $\delta$ is the counterpart in the continuum approximation of the mean imbalance

\begin{equation} \label{imbalance}
\langle\chi_{34}| \, \frac{a_3^\dag a_3-a_4^\dag a_4}{\nu} \, |\chi_{34}\rangle
\end{equation}
that may depend on $\nu$, and takes into account the fact that the coefficients $x_k$ can be picked on any Fock state. We now estimate the integral in equation (\ref{fid.int}). Rescaling $\chi(z)=\sqrt{\alpha}\zeta((z+\delta)\alpha)$, we compute

\begin{align}
& \nu\int_{-1}^1 dz\int_{-1}^1 dy \max\left\{0,N+1-|z-y|\frac{\nu}{2}\right\}\bar\chi(z)\chi(y) \nonumber \\
& =\frac{\nu}{\alpha(\nu)}\int_{\alpha(\delta-1)}^{\alpha(1+\delta)} dz'\int_{\alpha(\delta-1)}^{\alpha(1+\delta)} dy'\max\left\{0,N+1-|z'-y'|\frac{\nu}{2\alpha}\right\}\bar\zeta(z')\zeta(y'),
\end{align}

\noindent
where $z'=(z+\delta)\alpha$ and $y'=(y+\delta)\alpha$. With a change of variables, $z''=(z'+y')/2$, $y''=z'-y'$, and defining $\epsilon=\alpha(\nu)(N+1)/\nu$, the previous integral becomes

\begin{align} \label{int.int}
& \frac{\nu}{\alpha(\nu)}\int_{\alpha(\delta-1)}^{\alpha(1+\delta)} dz''\int_{2|z''-\alpha\delta|-2\alpha}^{2\alpha-2|z''-\alpha\delta|} dy''\max\left\{0,N+1-|y''|\frac{\nu}{2\alpha}\right\}\bar\zeta\left(z''+\frac{y''}{2}\right)\zeta\left(z''-\frac{y''}{2}\right) \nonumber \\
& =\frac{\nu}{\alpha(\nu)}\int_{\alpha(\delta-1)+\epsilon}^{\alpha(1+\delta)-\epsilon} dz''\int_{-2\epsilon}^{2\epsilon} dy''\left(N+1-|y''|\frac{\nu}{2\alpha}\right)\bar\zeta\left(z''+\frac{y''}{2}\right)\zeta\left(z''-\frac{y''}{2}\right) \nonumber \\
& +\frac{\nu}{\alpha(\nu)} \!\!\!\!\!\!\! \int\limits_{\alpha-\epsilon<|z''-\alpha\delta|<\alpha} \!\!\!\!\!\!\! dz'' \int_{2|z''-\alpha\delta|-2\alpha}^{2\alpha-2|z''-\alpha\delta|} dy''\left(N+1-|y''|\frac{\nu}{2\alpha}\right)\bar\zeta\left(z''+\frac{y''}{2}\right)\zeta\left(z''-\frac{y''}{2}\right).
\end{align}

Expanding the function $\zeta$ as $\zeta(z''\pm y''/2)=\zeta(z'')+\tilde R$ in (\ref{int.int}) and noting that

\begin{equation} \label{integral}
\int_{-b}^b dx (a-|x|) x^j=\frac{b \, (b^j+(-b)^j)(2a-b+aj-bj)}{j^2+3j+2}
\end{equation}
for $a\geq b\geq 0$ and $j\geq 0$, we can compute the integrals in $y''$. The dominant term comes form the first double integral in the right hand side of (\ref{int.int}) and is constant in $\epsilon$, while the rest of the right hand side of (\ref{int.int}) is of order $(N+1)^2{\cal O}(\epsilon)$. Computing the dominant term, the integral (\ref{int.int}) is

\begin{equation}
2(N+1)^2\left(\int_{\alpha(\delta-1)+\epsilon}^{\alpha(1+\delta)-\epsilon} dz''|\zeta(z'')|^2-R-{\cal O}(\epsilon)\right)=2(N+1)^2\left(1-R-{\cal O}(\epsilon)\right), \label{estimate}
\end{equation}

\noindent
where $R$ is the remainder originated form the first double integral in the right hand side of (\ref{int.int}). In equality \eqref{estimate}, we estimated the error between the remaining integral and the normalization:

\begin{equation}
1=\int_{\alpha(\delta-1)}^{\alpha(1+\delta)}dz''\zeta(z'')=\int_{\alpha(\delta-1)+\epsilon}^{\alpha(1+\delta)-\epsilon} dz''|\zeta(z'')|^2+{\cal O}(\epsilon).
\end{equation}
If $\zeta$ is continuous, then $\tilde R=o(1)$ and $R=o(\alpha(\nu))$, with $\lim_{y''\to 0}o(1)=0$ and, hence, $\lim_{\epsilon\to 0}o(1)=0$. If $\zeta$ is differentiable, then

\begin{equation}
\tilde R={\cal O}\big(y''\big)\Rightarrow\tilde R={\cal O}(\epsilon)\Rightarrow R=o\big(\alpha(\nu)\epsilon\big).
\end{equation}
If $\zeta$ is twice differentiable, then

\begin{equation}
\tilde R=y''\frac{d\zeta(z)}{dz}+{\cal O}\big(y''\big)^2 \Rightarrow \tilde R=y''\frac{d\zeta(z)}{dz}+{\cal O}(\epsilon^2)\Rightarrow R={\cal O}\big(\alpha(\nu)\epsilon^2\big).
\end{equation}
Plugging these estimations into equation (\ref{fid.int}), we prove the first equation in (\ref{fid.ent.asym}).

The second equation in (\ref{fid.ent.asym}) follows from the application of the first equation in (\ref{fid.ent.asym}) to the inequality \eqref{triangle} and the fact that $E$ cannot be larger than $\pi N/8$. The same result can be proven by a straightforward computation, as done for the first equation in (\ref{fid.ent.asym}). In equation (\ref{ent.int}), an integral similar to (\ref{int.int}) should be estimated, where $\zeta(\cdot)$ is replaced by $|\zeta(\cdot)|$. In these estimates, we have to notice that if $\zeta$ is continuous then $|\zeta|$ is continuous as well. However, if $\zeta$ is differentiable then $|\zeta|$ is no longer differentiable in the points where $\zeta$ crosses zero with non-vanishing derivative. These points are isolated and contribute with measure zero to the integral in (\ref{ent.int}).
\end{proof}

Hence, if the remainders in formulas (\ref{fid.ent.asym}) go to zero as $\nu\to\infty$, then the resource state provides asymptotically perfect teleportation.

\begin{remark} \label{rem1}
\emph{The physical meaning of the function $\delta(\nu)$ is that it proves the independence of Proposition \eqref{prop} from the mean imbalance \eqref{imbalance} between the modes, that represents the Fock state around which the superposition (\ref{chi34}) is centred. Proposition \ref{prop} implies that the convergence to the perfect teleportation does not depend on the mean imbalance \eqref{imbalance}.}
\end{remark}

\begin{remark} \label{rem2}
\emph{The crucial point in the above proof is the existence of a scaling such that $\zeta(z''\pm y''/2)=\zeta(z'')+o(1)$, with $\lim_{\nu\to\infty}o(1)=1$. One could directly impose this condition or the more general condition $\chi(z\pm y/2)=\chi(z)+o(1)$, without the rescaling $z\to(z+\delta(\nu))\alpha(\nu)$. However, the weaker conditions on the continuity or differentiability of $\zeta$, and the independence on $\nu$ are easier to be checked and to be exploited, to find states (\ref{chi34}) which provide asymptotically perfect teleportation.}
\end{remark}

\begin{proposition} \label{prop2}
If two states

\begin{equation}
|\chi_{34}^{(j)}\rangle=\sum_{k=0}^\nu x_k^{(j)}|k\rangle_3\otimes|\nu-k\rangle_4, \qquad j=1,2,
\end{equation}

\noindent
with non-negative coefficients $x_k^{(1,2)}\geq 0$, provide asymptotically perfect teleportation for $\nu\to\infty$, and turn out to be orthogonal in the same limit, $\lim_{\nu\to\infty}\langle\chi_{34}^{(2)}|\chi_{34}^{(1)}\rangle=0$, then their normalized non-negative superposition $|\chi_{34}\rangle=c_1|\chi_{34}^{(1)}\rangle+c_2|\chi_{34}^{(2)}\rangle$ with $c_1,c_2\geq 0$, provides asymptotically perfect teleportation.

In addition, if the states $|\chi_{34}^{(1,2)}\rangle$ satisfy proposition \ref{prop} with real functions $\delta_{1,2}(\nu)\in[-1,1]$ and $\alpha_{1,2}(\nu)$, respectively, the state $|\chi_{34}\rangle$ satisfies the same estimations (\ref{fid.ent.asym}) with $\max\{\alpha_1(\nu),\alpha_2(\nu)\}$ instead of $\alpha(\nu)$.
\end{proposition}

\begin{proof}
The teleportation fidelity of the state $|\chi_{34}\rangle$ is

\begin{equation}
f=\frac{1}{N+2}+\frac{\nu}{2}\sum_{k,j=0}^\nu \frac{\max\left\{0,N+1-|k-j|\right\}}{(N+1)(N+2)}\left(c_1^2 x_k^{(1)}x_j^{(1)}+c_2^2 x_k^{(2)}x_j^{(2)}+2c_1 c_2x_k^{(1)}x_j^{(2)}\right).
\end{equation}

\noindent
The first two contributions are the fidelities of the states $|\chi_{34}^{(1,2)}\rangle$, weighted with the coefficients $c_{1,2}^2$. The normalization of the state $|\chi_{34}\rangle$ implies

\begin{equation}
1=c_1^2+c_2^2+2c_1 c_2\langle\chi_{34}^{(1)}|\chi_{34}^{(2)}\rangle\xrightarrow[\nu\to\infty]{}c_1^2+c_2^2.
\end{equation}
Furthermore, the states $|\chi_{34}^{(1,2)}\rangle$ provide perfect teleportation in the same limit, and thus fidelity one. Putting all these considerations together, the asymptotic teleportation fidelity of $|\chi_{34}\rangle$ is

\begin{equation} \label{fid.sup}
\lim_{\nu\to\infty}f=\lim_{\nu\to\infty}\bigg(1+\nu c_1 c_2\sum_{k,j=0}^\nu \frac{\max\left\{0,N+1-|k-j|\right\}}{(N+1)(N+2)}x_k^{(1)}x_j^{(2)}\bigg).
\end{equation}

\noindent
The second term on the right hand side of (\ref{fid.sup}) is non-negative. Indeed, it vanishes in the limit $\nu\to\infty$, because the fidelity cannot exceed one, by definition. Therefore, we get $\lim_{\nu\to\infty}f=1$. Moreover, $\lim_{\nu\to\infty}E=\pi N/8$, because of the inequality \eqref{triangle}, and since the averaged final entanglement cannot exceed the value $\pi N/8$. The corrections to the asymptotic performances of $|\chi_{34}\rangle$ are at most of the same order of those of $|\chi_{34}^{(1,2)}\rangle$, because the superposition can only enhance the teleportation performances under the assumptions $c_{12},x_k^{(1,2)}\geq 0$.
\end{proof}

\begin{remark} \label{rem3}
\emph{The function $\alpha(\nu)$ carries information on the convergence towards perfect teleportation. In fact, if $N$ is fixed and $\nu\to\infty$, the difference between the actual teleportation performances and the asymptotic performances decreases with $\alpha(\nu)$.}
\end{remark}

Proposition \ref{prop} allows us to study the asymptotic performances of teleportation for several resource states. First, we check the consistency of proposition \ref{prop} with respect to some of the resource states discussed in \cite{Marzolino2015}: separable states, N00N states, and the maximally entangled state \eqref{max.ent.state}. Separable states do \emph{not} satisfy the continuity requirement for $\chi(z)$, since the values $x_k$ cannot be approximated with a continuous function. The same happens for superpositions of few Fock states, such as N00N states. The maximally entangled resource state (\ref{max.ent.state}) satisfies proposition \ref{prop} with $\alpha(\nu)=1$, $\delta(\nu)=0$ and $\chi(z)=1/\sqrt{2}$, providing asymptotically perfect teleportation. Indeed, we already derived this result from the analytical computations in \cite{Marzolino2015}. Propositions \ref{prop} and \ref{prop2} imply that teleportation endowed with the other resource states discussed in \cite{Marzolino2015} is asymptotically perfect, in accordance with the numerical computations presented there. It is instructive to give some examples of states which do not reach asymptotically perfect teleportation, and indeed do not satisfy the hypotheses of proposition \ref{prop}. Simple instances are the maximally entangled state (\ref{max.ent.state}) and the other resource states considered in \cite{Marzolino2015}, where additional phases $e^{i\vartheta(k)}$ multiply each Fock state in the superposition. The conditions of proposition \ref{prop} are not met if the phases scale differently from the moduli. Indeed, the teleportation performances computed numerically are far from their maximal values.

We now state some applications of proposition \ref{prop} to mixed resource states $\rho_{34}$ in (\ref{res}).

\begin{corollary} \label{cor1}
If there are two real functions $\delta(\nu)\in[-1,1]$ and $\alpha(\nu)$ such that the function $\omega(z,y)$ is rescaled as $\omega(z,y)=\alpha\xi((z+\delta)\alpha,(y+\delta)\alpha)$ and the function $\xi(\cdot,\cdot)$ does not depend on $\nu$, then the estimations (\ref{fid.ent.asym}) hold, depending on whether $\xi(\cdot,\cdot)$ is continuous, differentiable, or twice differentiable in $[\alpha(\delta-1),\alpha(1+\delta)]\times[\alpha(\delta-1),\alpha(1+\delta)]$.
\end{corollary}

\begin{proof}
The proof follows the same steps as proposition \ref{prop}, with the substitution $\bar\chi(z)\chi(y)\to\omega(z,y)$ and, subsequently, $\bar\zeta(z)\zeta(y)\to\xi(z,y)$.
\end{proof}

We observe that the rescaling $z\to(z+\delta(\nu))\alpha(\nu)$ is crucial only for the evaluation of the integrals in the variable $y''=(z-y)\alpha(\nu)$. Moreover, if the resource state is factorized in the variables $z+y$ and $z-y$, i.e $\omega(z,y)=\omega_+(z+y)\omega_-(z-y)$, the double integrals in (\ref{fid.int},\ref{ent.int}) can be factorized into products of single integrals. Examples are Gaussian states. This brings us to the following corollary.

\begin{corollary} \label{cor2}
If there is a real function $\alpha(\nu)$ such that $\omega(z,y)=\omega_+(z+y)\omega_-((z-y)\alpha)$, and $\omega_-(\cdot)$ does not depend on $\nu$, then the estimations (\ref{fid.ent.asym}) hold, depending on whether $\omega_-(\cdot)$ is continuous, differentiable or twice differentiable in $[-2\alpha,2\alpha]$.
\end{corollary}

\begin{proof}
The computation of the fidelity is similar to that performed in proposition \ref{prop}. With the change of variables from $(z,y)$ to $(z',y')=((z+y)/2,(z-y)\alpha)$, the integral in equation (\ref{fid.int}) becomes

\begin{align} \label{int.int.fact}
& \nu\int_{-1}^1 dz\int_{-1}^1 dy \max\left\{0,N+1-|z-y|\frac{\nu}{2}\right\}\omega(z,y) \nonumber \\
& =\frac{\nu}{\alpha(\nu)}\int_{\epsilon-1}^{1-\epsilon} dz'\omega_+(2z')\int_{-2\epsilon}^{2\epsilon} dy' \left(N+1-|y'|\frac{\nu}{2\alpha}\right)\omega_-(y') \nonumber \\
& +\frac{\nu}{\alpha(\nu)}\int\limits_{{1-\epsilon<|z''|<1}}dz' \omega_+(2z')\int_{2|z'|-2}^{2-2|z'|} dy' \left(N+1-|y'|\frac{\nu}{2\alpha}\right)\omega_-(y'),
\end{align}

\noindent
where $\epsilon=\alpha(\nu)(N+1)/\nu$. Then, all the estimations proceed as in proposition \ref{prop}, with the difference that we only need to apply the continuity or differentiability conditions to the function $\omega_-(\cdot)$, and use the normalization

\begin{equation}
\int_{-1}^1 dz \, \omega_+(2z)\omega_-(0)=\int_{-1}^1 dz \, \omega(z,z)=1.
\end{equation}
The computation of the average final entanglement follows from the inequality \eqref{triangle}, and from the fact that the maximum value of $E$ is $\pi N/8$.
\end{proof}

Since the teleportation protocol is linear in the resource state, we can use the previous propositions and corollaries to study the asymptotic teleportation performances of more general mixed resource states.

\begin{corollary} \label{cor3}
If a resource state is a mixture of states satisfying the hypotheses of propositions \ref{prop}, \ref{prop2} or of corollaries \ref{cor1}, \ref{cor2}, namely $\rho_{34}=\sum_i p_i\rho_{34}^{(i)}$, and if there are real functions $\delta_i(\nu)\in[-1,1]$ and $\alpha_i(\nu)$ as in (\ref{fid.ent.asym}) for each $\rho_{34}^{(i)}$, its asymptotic teleportation performances satisfy the equations (\ref{fid.ent.asym}), with $\max_i\alpha_i(\nu)$ instead of $\alpha(\nu)$.
\end{corollary}

\begin{proof}
The computation of the fidelity is directly implied by the linearity of the fidelity with respect to the resource state. The computation of the \eqref{triangle} and the fact that the maximum value of $E$ is $\pi N/8$.
\end{proof}

\subsection{Application: ground state of the double well potential}

The previous propositions allow us to study the asymptotic teleportation performances of resource states, even if the latter are not explicitly known. In this perspective, we now apply the previous propositions to analyse the asymptotic teleportation performances of the ground state of the two-mode Bose-Hubbard Hamiltonian with two-body interactions, as a resource state in the teleportation protocol. From the physical point of view, this resource state can be prepared with nowadays' technologies, such as magnetic traps and evaporative cooling \cite{Thomas2002}. The Bose-Hubbard Hamiltonian reads

\begin{equation} \label{BH}
H=-\tau\left(a_3^\dag a_4+a_4^\dag a_3\right)+U\left(n_3(n_3-1)+n_4(n_4-1)\right),
\end{equation}

\noindent
where $n_{3,4}=a^\dag_{3,4}a_{3,4}$, $\tau$ is the tunnelling amplitude between the wells of a double-well potential, and $U$ is the on-site interparticle interaction strength.

The ground state of (\ref{BH}) was studied in the limit of large particle numbers in \cite{Buonsante2012}. If

\begin{equation}
-1+\nu^{-\frac{2}{3}}\ll\gamma\equiv\frac{\nu U}{\tau}\ll\nu^2,
\end{equation}
the continuum approximation \eqref{chi34} of the ground state is a Gaussian superposition centred in $z=0$ with variance $\sigma_{\gamma}^2$:

\begin{equation} \label{cont.approx.1gauss}
\chi(z)=\frac{e^{-\frac{z^2}{(4\sigma_{\gamma}^2)}}}{(2\pi\sigma_{\gamma}^2)^{\frac{1}{4}}}, \qquad \sigma_{\gamma}^2=\frac{1}{\nu\sqrt{\gamma+1}}.
\end{equation}
Proposition \ref{prop} with $\alpha(\nu)=\nu^{1/2}$ and $\delta(\nu)=0$ implies that this Gaussian ground state provides asymptotically perfect teleportation when employed as a resource state.

If

\begin{equation}
-\sqrt{\nu}\ll\gamma\ll-1-\nu^{-\frac{2}{3}},
\end{equation}
the ground state is the superposition of two Gaussians:

\begin{equation} \label{cont.approx.2gauss}
\chi(z)=\frac{\chi_{-z_0}(z)+\chi_{z_0}(z)}{\sqrt{2}},
\end{equation}
with

\begin{equation}
\chi_{\pm z_0}(z)=\frac{e^{-\frac{(z\pm z_0)^2}{(4\sigma_{\gamma}'^2)}}}{(2\pi\sigma_{\gamma}'^2)^{\frac{1}{4}}}, \qquad \sigma_{\gamma}'^2=\frac{1}{(\nu|\gamma|\sqrt{\gamma^2-1})},
\end{equation}
and $z_0=\sqrt{1-1/\gamma^2}$. This ground state is the superposition of two Gaussian states each of which provides asymptotically perfect teleportation, as follows from proposition \ref{prop} with $\alpha(\nu)=\nu^{1/2}$ and $\delta(\nu)=\pm\sqrt{1-1/\gamma^2}$. These two states become orthogonal states when $\nu\to\infty$. Thus, the ground state satisfies the hypothesis of proposition \ref{prop2} and provides asymptotically perfect teleportation. An illustration of the two regimes is sketched in figure \ref{gaussians}.

\begin{figure}[htbp]
\centering
\label{fid.sym.coh.a}
\includegraphics[width=0.7\columnwidth]{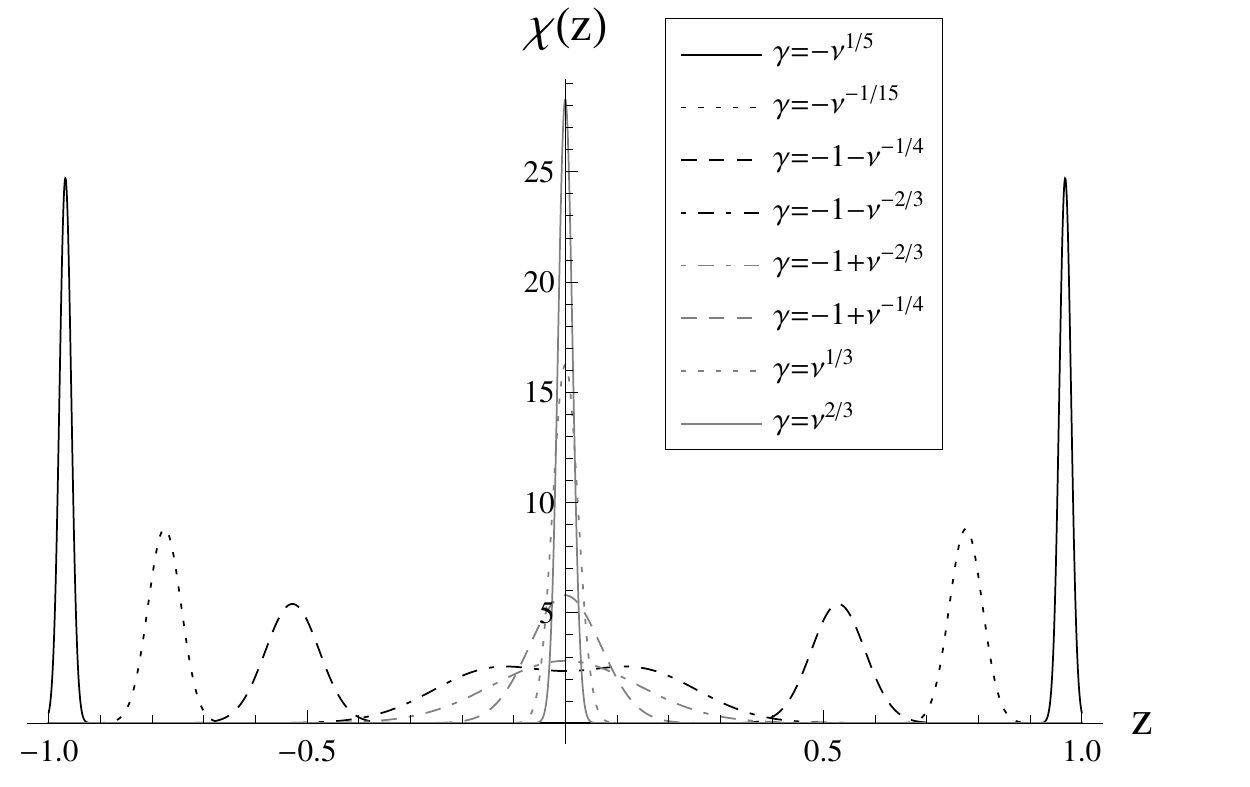}
\caption{Sketch of the continuum approximation $\chi(z)$, equations \eqref{cont.approx.1gauss} and \eqref{cont.approx.2gauss}, of the ground state of (\ref{BH}) in the regimes $-1+\nu^{-2/3}\ll\gamma\ll\nu^2$ (Gaussians) and $-\sqrt{\nu}\ll\gamma\ll-1-\nu^{-2/3}$ (superpositions of two Gaussians) for $\nu=1000$ and different values of $\gamma$: $\gamma=-\nu^{-1/5}$ (solid, black), $\gamma=-\nu^{-1/15}$ (dotted, black), $\gamma=-1-\nu^{-1/4}$ (dashed, black), $\gamma=-1-\nu^{-2/3}$ (dotdashed, black), $\gamma=-1+\nu^{-2/3}$ (dotdashed, gray), $\gamma=-1+\nu^{-1/4}$ (dashed, gray), $\gamma=\nu^{1/3}$ (dotted, gray), $\gamma=\nu^{2/3}$ (solid, gray).}
\label{gaussians}
\end{figure}

In the intermediate regime

\begin{equation}
-1-\nu^{-\frac{2}{3}}\ll\gamma\ll-1+\nu^{-\frac{2}{3}},
\end{equation}
the continuum approximation of the ground state $\chi(z)$ is picked around $z=0$, but it is not a Gaussian because it starts to feel the separation into two Gaussians, and its analytical expression is not known \cite{Buonsante2012}. Nevertheless, the ground state satisfies the hypothesis of proposition \ref{prop} with $\alpha(\nu)=\nu^{1/3}$ and $\delta(\nu)=0$ \cite{Buonsante2012}, thus provides asymptotically perfect teleportation. This is an application of proposition \ref{prop} that determines asymptotically perfect teleportation performances, though we do not explicitly know the resource state. According to propositions \ref{prop},\ref{prop2} and to the interpretation of $\alpha(\nu)$ given in remark \ref{rem3}, the teleportation performances of the ground state in the intermediate regime converge to the maximum values faster than the teleportation performances of other ground states, but slower than the teleportation performances of the maximally entangled state (\ref{max.ent.state}).

\subsection{Example: Gaussian states}

If the specific form of the resource state is known, one can directly compute the asymptotic teleportation performances (\ref{fid.int},\ref{ent.int}) and derive the exact errors from the asymptotic behaviour, as shown in the following example.

Let us consider a resource state whose continuum approximation \eqref{cont.appr.mixed} is

\begin{equation} \label{res.ex}
\omega(z,y)=\omega_+(z+y)e^{-\alpha^2(\nu)(z-y)^2},
\end{equation}
where the normalization reads $\int_{-1}^1 dz \, \omega_+(2z)\omega_-(0)=\int_{-1}^1 dz \, \omega(z,z)=1$. We follow the same steps as in the proof of proposition \ref{prop}, without the rescaling $z\to(z+\delta(\nu))\alpha(\nu)$. The expansion of the function $\omega$ becomes

\begin{equation}
\omega\left(z+\frac{y}{2},z-\frac{y}{2}\right)=\omega(z,z)\left(\sum_{j=0}^\infty \frac{y^{2j}\alpha^{2j}(\nu)}{j!}\right).
\end{equation}
Plugging this expansion into the equation of the fidelity and exploiting equation \eqref{integral} for $a\geq b\geq 0$ and $j\geq 0$, we estimate the fidelity as done in proposition \ref{prop}:

\begin{equation}
f=1-{\cal O}\left(\alpha^2(\nu)\frac{N^2}{\nu^2}\right)-{\cal O}\left((\omega(1,1)+\omega(-1,-1))\frac{N}{\nu}\right),
\end{equation}
where the second and the third term come from the first and the second double integral in the right hand side of (\ref{int.int}) respectively. The fidelity goes to one if and only if $\alpha^2(\nu)N^2/\nu^2\to 0$ and $\textnormal{and} \quad \big(\omega(1,1)+\omega(-1,-1)\big)N/\nu\to 0$. The inequality \eqref{triangle} implies the same conditions for the asymptotic average final entanglement.

The previous example recovers the maximally entangled state (\ref{max.ent.state}), the ground state of the Hamiltonian (\ref{BH}) in the Gaussian regime, and more general pure Gaussian states with

\begin{equation}
x_k=\frac{e^{-\frac{(k-k_0)^2}{4\sigma^2}}}{\sum_{k=0}^\nu e^{-\frac{(k-k_0)^2}{2\sigma^2}}}, \qquad \sigma\sim\nu^\beta, \qquad \beta>0.
\end{equation}
In this latter case, $\alpha(\nu)\sim\nu^{1-\beta}$ and $\omega(1,1)=\omega(-1,-1)\sim e^{-\nu^{2-2\beta}}$. The convergence is faster than that of the maximally entangled state (\ref{max.ent.state}) if

\begin{equation}
0=\lim_{\nu\to\infty}(1-f)\frac{\nu}{N}=e^{-\nu^{2-2\beta}}+N \nu^{1-2\beta},
\end{equation}
namely $1/2<\beta<1$. If $\beta\geq 1$,

\begin{equation}
\lim_{\nu\to\infty}(1-f)\frac{\nu}{N}=e^{-\nu^{2-2\beta}}+N \nu^{1-2\beta}=1
\end{equation}
and the convergence rate is the same as for the maximally entangled state. Strictly speaking, only the maximally entangled state can provide a probabilistic perfect teleportation, for finite $\nu$, while a Gaussian state always introduces a distortion of the teleported state.

\section{Robustness of the resource state and of the teleportation performances} \label{robustness}

As a further application of the previous properties, we discuss performances of the teleportation protocol, when the resource state $\rho_{34}$ is affected by noise. The noise is typically generated by dissipative dynamics due to the interaction with the environment, like a thermal bath or a lossy channel. We treat two different ways to model the noise.

\subsection{Mixing channel}

The first model consists in mixing the resource state with an undesired state

\begin{equation} \label{mixture}
\tilde\rho_{34}=\frac{\rho_{34}+s \sigma_{34}}{1+s},
\end{equation}
where the state is left unchanged with probability $1/(1+s)$ and is transformed into $\sigma_{34}$ with probability $s/(1+s)$. This description of the noise applies for instance when the Krauss operators of the noisy time-evolution \cite{NielsenChuang} are known, and the contribution which does not change the state can be singled out. Given the teleportation protocols ${\cal T}_{\rho,\sigma}$ and the fidelities $f_{\rho,\sigma}$ provided, respectively, by the resource states $\rho_{34}$ and $\sigma_{34}$, the average teleported state and the teleportation fidelity of the overall mixture $\tilde\rho_{34}$ are

\begin{equation}
\mathcal{T}_{\tilde\rho}\big[|\psi_{12}\rangle\langle\psi_{12}|\big]=\frac{\mathcal{T}_{\rho}\big[|\psi_{12}\rangle\langle\psi_{12}|\big]+s\mathcal{T}_{\sigma}\big[|\psi_{12}\rangle\langle\psi_{12}|\big]}{1+s}, \qquad f_{\tilde\rho}=\frac{f_{\rho}+sf_{\sigma}}{1+s}
\label{mix.fid}
\end{equation}

If the original resource state $\rho_{34}$ outperforms separable resource states $f_{\rho}>f_\textnormal{sep}$, and the state $\sigma_{34}$ is separable, $f_{\sigma}=f_\textnormal{sep}$, then $f_{\tilde\rho}>f_\textnormal{sep}$ for all finite $s$. Recall that from the inequality \eqref{triangle} with the right-hand-side being equal to $(N+2)(f-f_\textnormal{sep})$, $f_{\tilde\rho}>f_\textnormal{sep}$ implies that the resource state outperforms separable states also with respect to the average final entanglement $E_{\tilde\rho}>0$. This is a special feature of teleportation with identical particles. Indeed, for every state $\rho$ of distinguishable particles there is a separable state $\sigma$ and a finite mixing parameter $\tilde s$, such that the mixture $(\rho+s\sigma)/(1+s)$ with $s\geq\tilde s$ is separable \cite{Vidal1999}. Therefore, the overall mixture cannot outperform separable states as a resource for the teleportation protocol. On the other hand, a complete erasure of the entanglement via mixtures with separable states is never possible for two-mode states of identical particles \cite{Benatti2012-2}, and the residual entanglement, as well as entanglement of the original state, are useful for teleportation.

A different situation occurs when the undesired state $\sigma$ is entangled. In this case, the complete erasure of the entanglement of $\rho_{34}$ via the mixture $\tilde\rho_{34}$ is possible for identical particles, as well as for distinguishable particles \cite{Steiner2003}. However, the complete erasure occurs under very precise conditions on $s \, \sigma_{34}$: its off-diagonal entries in the Fock basis must erase the off-diagonal entries of $\rho_{34}$ in the same basis \cite{Benatti2012-2}, and therefore entanglement can be regenerated by small perturbations. These properties are reflected in the average final entanglement of the teleported state. It can happen that even if entanglement is not completely erased, the resulting mixture $\tilde\rho_{34}$ does not outperform the teleportation fidelity of separable states, since the fidelities $f_{\tilde\rho,\rho,\sigma}$ can be smaller than $f_\textnormal{sep}$. However, the linearity of the teleportation protocol implies that if both the fidelities $f_\rho,f_\sigma$ improve over $f_\textnormal{sep}$, then $f_{\tilde\rho}$ also does. Furthermore, $f_{\tilde\rho}>f_\textnormal{sep}$ implies $E_{\tilde\rho}>0$, due to \eqref{triangle}. If both $\rho_{34}$ and $\sigma_{34}$ provide asymptotically perfect teleportation, the same happens to $\tilde\rho_{34}$, as a consequence of corollary \ref{cor3}.

\subsection{Master equations}

The second description of noisy dynamics consists of master equations which generate time-evolutions of the system \cite{BreuerPetruccione,BenattiFloreanini}. The two most relevant sources of noise in the context of ultracold atoms are dephasing and particle losses \cite{Leibfried2003}. Entanglement affected by these dynamics feels an exponential damping, which goes to zero only asymptotically in time \cite{Argentieri2011,Marzolino2013}. Therefore, these systems do not experience finite-time disentanglement which is a generic behaviour of distinguishable particles affected by local noisy dynamics \cite{Zyczkowski2001,Carvalho2004,Fine2005}. This feature affects the performances of teleportation when the resource state (\ref{res}) undergoes noisy dynamics.

\subsubsection{Dephasing}

Let's start with the dephasing described by the following Markovian master equation

\begin{equation} \label{deph}
\frac{d}{dt}\rho_{34}(t)=\sum_{i=3,4} \lambda_i\left(a_i^\dag a_i\rho_{34}(t) a_i^\dag a_i-\frac{1}{2}\left\{(a_i^\dag a_i)^2,\rho_{34}(t)\right\}\right),
\end{equation}

\noindent
with positive constants $\lambda_i$. The solution \cite{Argentieri2011} reads

\begin{equation} \label{sol.deph}
(\rho_{34})_{k,j}(t)=e^{-\frac{t}{2}(\lambda_3+\lambda_4)(k-j)^2}(\rho_{34})_{k,j}(0),
\end{equation}

\noindent
see also \cite{Benatti2012,Marzolino2013} for generalizations to many modes. A resource state affected by dephasing can lose its capability to outperform separable states. Let us consider the initial state

\begin{equation} \label{lose}
\rho_{34}(0)=
\begin{pmatrix}
\begin{matrix}
a & 0 & 0 & y \\
0 & b & x & 0 \\
0 & x & c & 0 \\
y & 0 & 0 & d
\end{matrix} & {\bf 0}_{4,\nu-3} \\
{\bf 0}_{\nu-3,4} & {\bf 0}_{\nu-3,\nu-3}
\end{pmatrix}
\end{equation}

\noindent
in the Fock basis $\{|k,\nu-k\rangle\}_k$, where $x<0$, $y>0$, and ${\bf 0}_{n,m}$ is the $n\times m$ matrix with all zero entries. The normalization reads $1=\textnormal{tr}(\rho_{34}(0))=a+b+c+d$, and the positivity of the state implies $x^2\leq bc$ and $y^2\leq ad$. If $N>2$, the fidelity (\ref{fidelity}) is larger than the one of separable states if and only if $y>-x \, \frac{N}{N-2}$. On the other hand, the evolved state (\ref{sol.deph}) outperforms the fidelity of separable states if and only if $y>-x \, e^{4t(\lambda_3+\lambda_4)}\frac{N}{N-2}$. Therefore, there are values of $y$ such that the initial state does better than the teleportation fidelity provided by separable states, but the evolved state at finite times does not. However, the resource states whose entries $(\rho_{34})_{k,j}(0)$ with $|k-j|<N+1$ are non-negative preserve their capability to outperform separable states at any finite time, when they are affected by dephasing. This class of states includes the states discussed in \cite{Marzolino2015}. Indeed, the positivity of the entries $(\rho_{34})_{k,j}(t)$ with $|k-j|<N+1$ ensures that the additional contribution to the fidelity (\ref{fidelity}) with respect to separable states is positive. Moreover, these off-diagonal entries are exponentially damped, but vanish only asymptotically in time. If the resource state at time zero is entangled, it remains entangled at any finite time \cite{Argentieri2011,Benatti2012,Marzolino2013}. Thus, the average final entanglement is strictly positive at any finite time if it is for the initial resource state, since the expression (\ref{av.ent}) only involves the modulus of the entries of the resource state.

Another question is whether the noisy resource state preserves the property to achieve asymptotically perfect teleportation in the limit $\nu\to\infty$. If the initial state is (\ref{res.ex}), the evolved state has the same form, with the substitution $\alpha^2(\nu)\to\alpha^2(\nu)+t(\lambda_3+\lambda_4)\nu^2/8$. The initial resource state provides asymptotically perfect teleportation, if and only if $\displaystyle \lim_{\nu\to\infty}\alpha^2(\nu)N^2/\nu^2=0$ and $\displaystyle \lim_{\nu\to\infty}\big(\omega(1,1)+\omega(-1,-1)\big)N/\nu=0$, while the evolved resource state does the same if and only if $\displaystyle \lim_{\nu\to\infty}t(\lambda_3+\lambda_4)N^2=0$.

\subsubsection{Particle loss}

Another relevant noise source that affects systems of ultracold atoms is particle loss, and can be described by the following general Markovian master equation

\begin{equation} \label{loss}
\frac{d}{dt}\rho_{34}(t)=\sum_i \lambda_i\left(A_i\rho_{34}(t) A_i^\dag-\frac{1}{2}\left\{A_i^\dag A_i,\rho_{34}(t)\right\}\right),
\end{equation}

\noindent
where the $\lambda_i$ are positive constants, and the $A_i$ are monomials of the annihilation operators $a_3$, $a_4$: $A_i=a_3^{m_i}a_4^{n_i}$. These assumptions generalize the two- and $m$-particle losses $A_i\in\{a_3 a_4,a_3^m,a_4^m\}$ discussed in \cite{Li2009,Pawlowski2010,Kepesidis2012}. The solution of (\ref{loss}) reads \cite{Marzolino2013}

\begin{equation} \label{sol.loss}
\rho_{34}(t)=e^{-t\sum_i\frac{\lambda_j}{2}A_i^\dag A_i}\rho_{34}(0)e^{-t\sum_i\frac{\lambda_i}{2}A_i^\dag A_i}+\sum_{\nu'=0}^{\nu-1}\rho_{34}^{(\nu')}(t),
\end{equation}

\noindent
where $\rho_{34}^{(\nu')}(t)$ are matrices which operate on the subspace of $\nu'<\nu$ particles. Each operator $A_i^\dag A_i$ is diagonal in the Fock basis, thus the solution (\ref{sol.loss}) has the explicit form

\begin{equation} \label{sol.loss2}
(\rho_{34})_{k,j}(t)=e^{-t(\eta_k+\eta_j)}(\rho_{34})_{k,j}(0)+\sum_{\nu'=0}^{\nu-1}(\rho_{34}^{(\nu')})_{k,j}(t), \qquad
\eta_k=\sum_i\frac{\lambda_i}{2}\frac{k!(\nu-k)!}{(k-m_i)!(\nu-k-n_i)!},
\end{equation}

\noindent
where we adopted the convention that $(k-m_i)!=\infty$ if $k-m_i<0$, and similarly for $(\nu-k-n_i)!$. We notice that the components $\rho_{34}^{(\nu'<\nu)}(t)$ of the noisy resource state do not contribute to the fidelity. Indeed, the fidelity is the overlap between the state to be teleported and the average teleported state. The state which is teleported by means of the component $\rho_{34}^{(\nu'<\nu)}(t)$ has $N-\nu+\nu'$ particles, and thus has vanishing overlap with the initial state which has $N$ particles. However, these components contribute to the average final entanglement, if $\nu-\nu'<N$. Since only the component of the mixture (\ref{sol.loss}) with $\nu$ particles contributes to the fidelity, the fidelity is damped by its weight $\sum_k e^{-2t\eta_k}(\rho_{34})_{k,k}(0)$. The fidelity of the evolved resource state is

\begin{equation}
f_{\rho_t}=\tilde f\sum_k e^{-2t\eta_k}(\rho_{34})_{k,k}(0),
\end{equation}
where $\tilde f$ is the fidelity provided by the normalized state

\begin{equation}
\frac{e^{-t\sum_i\frac{\lambda_j}{2}A_i^\dag A_i}\rho_{34}(0)e^{-t\sum_i\frac{\lambda_i}{2}A_i^\dag A_i}}{\sum_k e^{-2t\eta_k}(\rho_{34})_{k,k}(0)}.
\end{equation}
In particular, estimating all the rates $\eta_k\leq\max_k\eta_k$, and due to the linearity of the fidelity, one gets the inequality

\begin{equation}
f_{\rho_t}\geq e^{-2t\max_k\eta_k}f_{\rho_0},
\end{equation}
where $f_{\rho_0}$ is the fidelity provided by the resource state at time zero. Then, if

\begin{equation}
2 \, t \, \max_k\eta_k<\ln\left(f_{\rho_0}\frac{N+2}{2}\right),
\end{equation}
the evolved resource state preserves the capability to outperform the separable states. Furthermore, the off-diagonal entries and, hence, the entanglement of the resource state, are exponentially damped, but vanish only asymptotically in time \cite{Argentieri2011,Benatti2012,Marzolino2013}. Thus, the average final entanglement is positive at any finite time, since the expression (\ref{av.ent}) involves only the modulus of the entries of the resource state.

Let us now investigate whether the noisy resource state preserves the property to achieve asymptotically perfect teleportation as $\nu\to\infty$. The presence of non-negligible $\rho_{34}^{(\nu'<\nu)}(t)$ prevents asymptotically perfect teleportation, because these matrix elements do not contribute to the fidelity. Thus, a necessary condition for the asymptotically perfect teleportation is the vanishing of $\rho_{34}^{(\nu'<\nu)}(t)$ when $\nu\to\infty$. These contributions are weighted with the probability

\begin{equation}
\sum_{\nu'<\nu}\textnormal{tr}(\rho_{34}^{(\nu')}(t))=1-\sum_{k=0}^\nu e^{-2t\eta_k}(\rho_{34})_{k,k}(0),
\end{equation}
where we used the property that the state $\rho_{34}(t)$ has trace one. Since the initial state is positive and normalized, $\sum_{k=0}^\nu(\rho_{34})_{k,k}(0)=1$, the weight $\sum_{\nu'<\nu}\textnormal{tr}(\rho_{34}^{(\nu')}(t))$ goes to zero if and only if time and the decay rates scale with $\nu$ such that $\lim_{\nu\to\infty}t\eta_k=0$, for all $k$. As a concrete example, we consider two-particle losses: $\{\lambda_j A_j\}_j=\{\lambda_3 a_3,\lambda_4 a_4,\lambda_{33}a_3^2,\lambda_{44}a_4^2,\lambda_{34}a_3 a_4\}$. The decay rates in (\ref{sol.loss}) are

\begin{align}
\eta_k+\eta_j= & \lambda_{44}\nu(\nu-1)-\lambda_4\nu+\frac{k+j}{2}(\lambda_3-\lambda_4-\lambda_{33}+\lambda_{44}(1-2\nu)+\nu\lambda_{34}) \nonumber \\
& +\left(\frac{(k+j)^2}{4}+\frac{(k-j)^2}{4}\right)(\lambda_{33}+\lambda_{44}-\lambda_{34}).
\end{align}

\noindent
The exponential damping $e^{-t(\eta_k+\eta_j)}$ in (\ref{sol.loss2}) factorizes in a product of an exponential function depending only on $k+j$, and the exponential $e^{-t(\lambda_{33}+\lambda_{44}-\lambda_{34})(k-j)^2/4}$. Thus, for the initial state (\ref{res.ex}), the evolved state has the same form with the substitution $\alpha^2(\nu)\to\alpha^2(\nu)+t(\lambda_{33}+\lambda_{44}-\lambda_{34})\nu^2/16$. Moreover, given an initial state that provides asymptotically perfect teleportation, namely $\displaystyle \lim_{\nu\to\infty}\alpha^2(\nu)N^2/\nu^2=0$ and $\displaystyle \lim_{\nu\to\infty}\big(\omega(1,1)+\omega(-1,-1)\big)N/\nu=0$, the above necessary conditions $\displaystyle \lim_{\nu\to\infty}t\eta_k=0$ to preserve the asymptotic teleportation performances turn out to be sufficient as well.

\section{Conclusions} \label{discussions}

We extended the generalization of quantum teleportation to identical massive particles. We first reviewed how the notion of entanglement and the teleportation protocol change in the presence of identical particles. Since particles cannot be individually addressed, we have to identify subsystems with subalgebras of observables that can be experimentally manipulated. Therefore, local parties do not own particles but rather orthogonal modes, such as in optical lattices with wells which can be split into groups. The aim of the teleportation protocol is to send the state of one mode to a mode owned by a receiver, by means of local operations, classical communication, and the aid of an entangled shared state. If the mode to be teleported is entangled with another mode, teleportation is also called entanglement swapping. One can divide a very long distance into segments, such that the noise within each segment is controllable. Swapping the entanglement across each segment, it is possible to share entanglement at distances along which the noise is not directly controllable. Thus, entanglement swapping can be used to share long-distance entanglement, as required for quantum networks \cite{Kimble2008}, without the need to physically gather the subsystems in the same place.

Perfect teleportation with identical particles is possible only when the number of particles in the resource state tends to infinity. Therefore, we derived sufficient conditions for the resource state to provide asymptotically perfect teleportation, in the above limit. These results generalize the examples explicitly discussed in \cite{Marzolino2015}. Moreover, they can be used to establish asymptotic teleportation performances even with only partial knowledge of the resource state. This situation was exemplified with the ground states of the double-well potential with two-body interactions, which can be prepared with available techniques, i.e. magnetic traps and evaporative cooling \cite{Thomas2002}.

Furthermore, we studied the robustness of teleportation performances against noise, in connection with the robustness of entanglement of the resource state. In order to model the noise, we considered the mixture of the resource state with undesired states, and the open system dynamics describing dephasing and particle losses. The capability to outperform separable states is preserved for any mixture with separable states and with almost all entangled states. The same holds true when some resource states undergo the above dissipative dynamics for any finite interval of time. The property to achieve asymptotically perfect teleportation is more fragile against noise.

\end{document}